\documentclass[11pt,dvipsnames]{article}
\usepackage[fleqn]{amsmath}
\usepackage[T1]{fontenc}
\usepackage[utf8]{inputenc} 
\usepackage{xcolor}
\usepackage{gitinfo}
\usepackage[margin=1in]{geometry}
\usepackage{complexity}
\usepackage{mathtools,amssymb}
\usepackage{setspace}
\usepackage{xspace}
\RequirePackage[colorlinks=true]{hyperref}
\hypersetup{
  linkcolor=[rgb]{0,0,0.5},
  citecolor=[rgb]{0, 0.5, 0},
  urlcolor=[rgb]{0.5, 0, 0}
}
\usepackage{todonotes}
\usepackage{dsfont}
\usepackage{mathpazo}

\usepackage{rpmacros} 

\usepackage{amsthm}
\usepackage{thmtools,thm-restate}

\numberwithin{equation}{section}

\declaretheoremstyle[bodyfont=\it,qed=\qedsymbol]{noproofstyle}

\declaretheorem[numberlike=equation]{observation}

\declaretheorem[name=Observation,numbered=no]{observation*}

\declaretheorem[numberlike=equation]{theorem}

\declaretheorem[name=Theorem,numbered=no]{theorem*}

\declaretheorem[numberlike=equation]{lemma}
\declaretheorem[name=Lemma,numbered=no]{lemma*}

\declaretheorem[numberlike=equation]{corollary}
\declaretheorem[name=Corollary,numbered=no]{corollary*}

\declaretheorem[name=Proposition,numbered=no]{proposition*}

\declaretheorem[name=Claim,numbered=no]{claim*}

\declaretheorem[name=Conjecture,numbered=no]{conjecture*}

\declaretheorem[name=Question,numbered=no]{question*}

\declaretheoremstyle[bodyfont=\it,qed=$\lozenge$]{defstyle} 

\declaretheorem[numberlike=equation,style=defstyle]{definition}
\declaretheorem[unnumbered,name=Definition,style=defstyle]{definition*}

\declaretheorem[unnumbered,name=Example,style=defstyle]{example*}

\declaretheorem[unnumbered,name=Notation=defstyle]{notation*}

\declaretheorem[unnumbered,name=Construction,style=defstyle]{construction*}

\declaretheorem[numberlike=equation,style=defstyle]{remark}
\declaretheorem[unnumbered,name=Remark,style=defstyle]{remark*}

\usepackage{nth}
\usepackage{intcalc}
\usepackage{etoolbox}
\usepackage{xstring}

\usepackage{ifpdf}
\ifpdf
\else
\usepackage[quadpoints=false]{hypdvips}
\fi

\newcommand{\ECCCurl}[2]{http://eccc.hpi-web.de/report/\ifnumcomp{#1}{>}{93}{19}{20}#1/#2/}

\newcommand{\shortECCC}[2]{\texttt{\href{http://eccc.hpi-web.de/report/\ifnumcomp{#1}{>}{93}{19}{20}#1/#2/}{eccc:TR#1-#2}}}

\newcommand{\parseECCC}[1]{
\StrSubstitute{#1}{TR}{}[\tmpstring]%
\IfSubStr{\tmpstring}{/}{ 
\StrBefore{\tmpstring}{/}[\ecccyear]%
\StrBehind{\tmpstring}{/}[\ecccreport]%
}{
\StrBefore{\tmpstring}{-}[\ecccyear]%
\StrBehind{\tmpstring}{-}[\ecccreport]%
}%
\shortECCC{\ecccyear}{\ecccreport}}

\makeatletter
\renewcommand{\vec}[1]{\mathbf{#1}}
\newcommand{\vx}{{\vec{x}}\@ifnextchar{^}{}{}}		

\DeclareMathOperator{\Var}{Var}

\def\epsilon{\varepsilon}

\allowdisplaybreaks 

\newcommand{\vSPSP}[1]{\Sigma\Pi\inparen{\Sigma\Pi}_{#1}}
\newcommand{\ehref}[1]{\texttt{\href{mailto:#1}{#1}}}

\date{}

\title{Towards Optimal Depth Reductions for Syntactically Multilinear Circuits}
\author{Mrinal Kumar \thanks{\ehref{mrinalkumar08@gmail.com}. Department of Computer Science, University of Toronto, Canada.  A part of this work was done  during the postdoctoral stay at Harvard, during the lower bounds semester at Simons Institute for the Theory of Computing, Berkeley and while visiting  TIFR, Mumbai.} \and Rafael Oliveira \thanks{\ehref{rafael@cs.toronto.edu}. Department of Computer Science, University of Toronto, Toronto, Canada. Part of this work was done while visiting the Simons Institute for 
the Theory of Computing.}
 \and Ramprasad Saptharishi \thanks{\ehref{ramprasad@tifr.res.in}. Tata Institute of Fundamental Research, Mumbai, India. Research supported by Ramanujan Fellowship of DST.}}
\onehalfspace
\begin{document}
\maketitle

{\let\thefootnote\relax
\footnotetext{\textcolor{white}{Git info: (\gitAuthorIsoDate)\;,\;\gitAbbrevHash\;\; \gitVtag}}
}
\onehalfspace

\begin{abstract}
We show that any $n$-variate polynomial computable by a syntactically multilinear circuit of size $\poly(n)$ can be computed by a depth-$4$ syntactically multilinear ($\Sigma\Pi\Sigma\Pi$) circuit of size at most $\exp\left({O\left(\sqrt{n\log n}\right)}\right)$. For degree $d = \omega(n/\log n)$, this improves upon the upper bound of $\exp\left({O(\sqrt{d}\log n)}\right)$ obtained by Tavenas~\cite{T15} for general circuits, and is known to be asymptotically optimal in the exponent when $d < n^{\epsilon}$ for a small enough constant $\epsilon$. Our upper bound matches the lower bound of $\exp\left({\Omega\left(\sqrt{n\log n}\right)}\right)$ proved by Raz and Yehudayoff~\cite{RY09}, and thus cannot be improved further in the exponent. Our results hold over all fields and also generalize to circuits of small individual degree.  

More generally, we show that an $n$-variate polynomial computable by a syntactically multilinear circuit of size $\poly(n)$ can be computed by a syntactically multilinear circuit of product-depth $\Delta$ of size at most $\exp\inparen{O\inparen{\Delta \cdot (n/\log n)^{1/\Delta} \cdot \log n}}$. It follows from the lower bounds of Raz and Yehudayoff~\cite{RY09} that in general, for constant $\Delta$, the exponent in this upper bound is tight and cannot be improved to $o\inparen{\inparen{n/\log n}^{1/\Delta}\cdot \log n}$. 
\end{abstract}

\thispagestyle{empty}
\pagenumbering{arabic}

\section{Introduction}
An algebraic circuit over a field $\F$ and variables $\vx = (x_1, x_2, \ldots, x_n)$ is a directed acyclic graph whose internal vertices (called gates) are labeled as either $+$ (sum) or $\times $ (product), and leaves (vertices of indegree zero) are labeled by the variables in $\vx$ or constants from $\F$. The gates of outdegree zero in a circuit are called its output gates. Algebraic circuits give a natural and succinct representation  for multivariate polynomials; analogous to the way Boolean circuits give a succinct representation of Boolean functions. We refer the reader to the excellent survey of Shpilka and Yehudayoff~\cite{SY10} for an introduction to the area of algebraic circuit complexity. One of the main protagonists in the results in this paper will be the class of syntactically multilinear circuits which we now define. 
\begin{definition}[Syntactically Multilinear Circuits ]
An algebraic circuit $C$ is said to be syntactically multilinear if at every product gate $v$ in $C$ with inputs $u_1, u_2, \ldots, u_t$, the set of variables in the sub-circuits rooted at $u_i$ are pairwise disjoint from each other.
\end{definition}

The size of an algebraic circuit is the number of edges in it, and its depth is the length of the longest path from an output gate to a leaf. Intuitively, the size of a circuit is an indicator of   the time complexity of computing the polynomial, and its depth indicates how fast the polynomial can be computed in parallel.

We now introduce a sequence of fundamental structural results for algebraic circuits, that are collectively called depth reductions; this is the main focus of this paper. 
\paragraph*{Depth Reductions. } In a beautiful, surprising and influential work, Valiant et al.~\cite{vsbr83} showed that every polynomial family which is efficiently computable by an algebraic circuit is also efficiently computable in parallel. Formally, they showed the following theorem. 
\begin{theorem}[\cite{vsbr83}]\label{thm:vsbr}
There is an absolute constant $c \in \N$ such that the following is true.  
If $P$ be an $n$-variate homogeneous polynomial of degree $d$ over any field $\F$ which can be computed by an algebraic circuit $C$ of size $s$, then $P$ can be computed by an algebraic circuit $C'$ (of unbounded fan-in) of depth $c\log d$ and size $(snd)^c$.
\end{theorem}
In particular, the theorem says that every polynomial family of polynomially bounded (in $n$) degree that is computable by a circuit of size $\poly(n)$ and arbitrary depth, is also efficiently computable by a circuit of size $\poly(\log n)$ and depth $O(\log n)$.

In a remarkable extension of~\autoref{thm:vsbr}, Agrawal and Vinay~\cite{av08} showed that one can parallelize algebraic circuits even more (reducing the depth to a constant), at the cost of a larger (a non-trivial subexponential  factor) blow up in the circuit size. The version of their theorem stated below is due to Tavenas~\cite{T15}, who optimized the parameters further. 
\begin{theorem}[\cite{av08, K12b, T15}]\label{thm:av}
There is an absolute constant $c \in \N$ such that the following is true.  
If $P$ is an $n$-variate homogeneous polynomial of degree $d$ over any field $\F$ which can be computed by an algebraic circuit $C$ of size $s$, then $P$ can be computed by a homogeneous $\Sigma\Pi\Sigma\Pi$ algebraic circuit $C'$ of size $(snd)^{c\sqrt{d}}$.
\end{theorem}

Here, a $\Sigma\Pi\Sigma\Pi$ circuit is an algebraic circuit with four layers of alternating sum and product gates with the top layer being a sum layer. Throughout this paper, when we say a depth-$4$ circuit, we mean a $\Sigma\Pi\Sigma\Pi$ circuit. 

We note that while~\autoref{thm:av} as stated above reduces a homogeneous circuit of arbitrary depth to a homogeneous circuit of depth-$4$, but it easily follows from the proof that the depth reduction preserves syntactic restrictions. That is, if we start with a syntactically multilinear and homogeneous circuit, the resulting depth-$4$ circuit is also syntactically multilinear and homogeneous. This statement will be of particular interest as we study depth reductions for syntactically multilinear circuits in this paper. 


\paragraph*{On the optimality of reductions to depth-$4$. }An immediate consequence of~\autoref{thm:vsbr} and~\autoref{thm:av} is that strong enough lower bounds for algebraic circuits of bounded depth imply superpolynomial lower bounds for general algebraic circuits. Thus, the questions of proving lower bounds for bounded depth circuits, and that of understanding if the parameters in~\autoref{thm:av} can be improved further seem to be of fundamental interest. In the last few years, we have had significant progress on both these fronts. Following a long line of work starting with a work of Kayal~\cite{K12a} and Gupta et al.~\cite{GKKS14}, we now know extremely good lower bounds for homogeneous depth-$4$ circuits.  
\begin{theorem}[Kumar and Saraf~\cite{KS14}]\label{thm:depth-4 lb}
There exists a polynomial family $\{f_n\}$, where $f_n$ is a homogeneous $n$-variate polynomial of degree $d = n^{\epsilon}$, for an absolute constant $\epsilon > 0$, such that $f_n$ is computable by an algebraic circuit of size $\poly(n)$, but any homogeneous depth-$4$ circuit computing $f_n$ has size $n^{\Omega(\sqrt{d})}$. 

Moreover, the family $\{f_n\}$ is computable by a syntactically multilinear circuit of polynomial size. 
\end{theorem}
If we allow the hard polynomial to be explicit but not necessarily have small circuits, then  upper bound on the degree $d$ in the above theorem can be increased to as large as $n^{1-\epsilon}$ for any constant $\epsilon > 0$.\footnote{Though this is not explicitly mentioned in these results, the proofs can be extended to this regime of parameters.} 
Thus, in general, the exponent in the upper bound on the size of the depth-$4$ circuit obtained in~\autoref{thm:av} cannot be improved asymptotically. In fact, the  theorem shows that we cannot even expect such an improvement for syntactically multilinear circuits in the setting when the degree  $d$ is sufficiently smaller than the number of variables $n$. A natural question here is to understand if~\autoref{thm:av} is also asymptotically tight in the exponent when the degree is larger. The following result of Raz and Yehudayoff goes a long way towards answering this question. 
\begin{theorem}[\cite{RY09}]\label{thm:ry ml constant depth}
  There is a family of multilinear polynomials $\{f_n\}$ such that, for every $n$, the polynomial $f_n$ is an $n$-variate degree $d = \Theta(n)$ polynomial that can be computed by a syntactically multilinear circuit of size $\poly(n)$, but any multilinear circuit of depth-$4$ computing $f_n$ has size $n^{\Omega\left(\sqrt{n/\log n}\right)}$.

  More generally, for any constant $\Delta$, any syntactically multilinear circuit of product-depth\footnote{Also referred to as a syntactically multilinear $\left(\Sigma\Pi\right)^{\Delta}$ circuit.} $\Delta$  computing $f_n$ must have size
  $n^{\Omega\inparen{(n/\log n)^{1/\Delta}}}$. 
\end{theorem}

For depth-$4$ circuits (or $\Delta = 2$), asimilar result was proved by Hegde and Saha~\cite{HS17} for the more general\footnote{A multilinear circuit is a multi-$k$-ic circuit for $k = 1$.} class of circuits called multi-$k$-ic circuits, where the \emph{formal degree} of any variable in the circuit is bounded by a parameter $k$ (formally defined in \autoref{defn:multi-k-ic}).

\begin{theorem}[\cite{HS17}]\label{thm:hegde-saha} There is an explicit family $\set{f_n}$ of $n$-variate multilinear polynomials of degree $d = \Theta(n)$ such that, for every $k \leq (n\log n)^{0.9}$, any multi-$k$-ic circuit of depth-$4$ computing $f_n$ has size at least $n^{\Omega\inparen{\sqrt{n/(k\log n)}}}$. 


\end{theorem}

Thus, \autoref{thm:ry ml constant depth} and \autoref{thm:hegde-saha} shows that the exponent $\sqrt{d}$ in the exponent in~\autoref{thm:av} cannot be replaced by $o\inparen{\sqrt{n/\log n}}$. Thus, in the regime when $d = \Theta(n)$, there is a gap of $\sqrt{\log n}$ between the known lower bounds and what is potentially achievable via depth reduction. Raz and Yehudayoff~\cite{RY09} also observe  that using their techniques, the lower bound cannot be improved to $n^{\omega(\sqrt{n/\log n})}$. Our main motivation for this work was to bridge this gap. In the light of~\autoref{thm:depth-4 lb}, we believed the upper bound of $n^{O(\sqrt{d})}$ in~\autoref{thm:av} to be right bound for multilinear circuits for all $d$, and had hoped to improve the lower bound in~\autoref{thm:ry ml constant depth} to $n^{\Omega(\sqrt{n})}$. 

However,  as we discuss next, the correct exponent for depth reduction to depth-$4$ in the high degree regime turns out to be $\sqrt{n/\log n}$. In addition to being surprising, this also offers a potentially viable approach to the question of proving superpolynomial  lower bounds for syntactically multilinear circuits by extending~\autoref{thm:depth-4 lb} to the high degree regime. We now state our results and discuss the connections to multilinear circuit lower bounds.
\subsection{Results}
We start by stating our main theorems.

\begin{restatable}{theorem}{multiKicDR}
\label{thm:multi-k-ic dr}
Let $C$ be a multi-$k$-ic circuit of size $s$ computing a polynomial in $n$ variables. Then, there is a multi-$k$-ic $\Sigma\Pi\Sigma\Pi$ circuit $C'$ of size $s^{O\left(\sqrt{\frac{kn}{\log s}}\right)}$ computing the same polynomial. 
\end{restatable}

\begin{restatable}{theorem}{multiKicDRDelta}
\label{thm:multi-k-ic dr Delta}
  Let $C$ be a multi-$k$-ic circuit of size $s$ computing a polynomial in $n$ variables. Then, there is a multi-$k$-ic $(\Sigma\Pi)^\Delta$ circuit $C'$ computing the same polynomial whose size is at most
  \[
    s^{O\inparen{\Delta \cdot (nk/\log s)^{1/\Delta}}}. 
  \]
\end{restatable}

Thus, for $s = \poly(n)$, $k = o(\log s)$ and $n \geq d \geq \omega\left(\frac{kn}{\log s}\right)$, the exponents in the upper bounds in \autoref{thm:multi-k-ic dr} are asymptotically better than that in~\autoref{thm:av}. An immediate consequence of \autoref{thm:multi-k-ic dr} is the following corollary. 
\begin{corollary}\label{cor:ml dr}
Let $\{f_n\}$ be an explicit family of multilinear polynomials, such that $f_n$ is an $n$ variate polynomial of degree $d = \omega(n/\log n)$, and any multilinear $\Sigma\Pi\Sigma\Pi$ circuit computing $f_n$ has size at least $n^{\Omega(\sqrt{d})}$. Then, $\{f_n\}$ requires superpolynomial size syntactically multilinear circuits. 
\end{corollary}
The corollary is of interest since by~\autoref{thm:depth-4 lb}, we know $n^{\Omega(\sqrt{d})}$ lower bounds for homogeneous multilinear $\Sigma\Pi\Sigma\Pi$ circuits, when $d = n^{\epsilon}$. Thus extending these bounds so that they hold for higher degree polynomials will imply superpolynomial lower bounds for multilinear circuits. The current best lower bound known for multilinear circuits is a nearly quadratic lower bound in a recent work of Alon et al.~\cite{AKV18}.
The standard technique for proving lower bounds for multilinear models is via the rank of the \emph{partial derivative matrix} under a random partition of variables (due to Raz~\cite{R09}).  This has been useful in almost all of the known lower bounds for multilinear models, such as super polynomial lower bounds for multilinear formulas~\cite{R09}, exponential lower bounds for constant depth multilinear circuits~\cite{RY09} as well as the currently known superlinear and nearly quadratic lower bounds for multilinear circuits~\cite{RSY08, AKV18}. However, this technique is too weak to yield even super-cubic lower bounds for syntactically multilinear circuits.
Thus, currently we do not even have potential approaches to proving superpolynomial lower bounds for multilinear circuits. In the light of this, it certainly seems worth exploring if the partial derivative based methods used in the proof of~\autoref{thm:depth-4 lb} can be extended to work for multilinear polynomials whose degree $d = \omega(n/\log n)$ is high. As far as we understand, there does not seem to be strong evidence one way or the other about this. 

For multi-$k$-ic circuits, we do not even know superpolynomial lower bounds for formulas or even constant depth formulas. Based on the discussion above,~\autoref{thm:multi-k-ic dr} does seem to offer a potentially viable approach to prove these lower bounds. 

Finally, we note again that the upper bound on the size of the depth-$4$ circuit obtained in~\autoref{thm:multi-k-ic dr} cannot be further improved asymptotically in the exponent as~\autoref{thm:ry ml constant depth} shows. 

  \subsection{Proof Overview}
We focus on giving an outline of the proof of~\autoref{thm:multi-k-ic dr} for the multilinear case (or $k = 1$). The proof follows the strategy of the proof of \autoref{thm:av} with some key differences, which we point out as we go along. There are two main steps and we now give an sketch of both of them.
\paragraph*{Balancing a syntactically multilinear circuit. } For this step, the key notion is that of a \emph{balanced} circuit. We say that a circuit $C$ is balanced with respect to a potential function $\Phi : C \rightarrow \N$ (e.g. degree, number of variables), if the fan-in of every product $g$ in $C$ is a constant, and $\Phi(g) \geq 2 \Phi(h)$ for every child $h$ of $g$. In the proof of~\autoref{thm:av}, the authors essentially use the results of Valiant et. al.~\cite{vsbr83} to balance a homogeneous circuit with the potential function $\Phi$ being the formal degree of a gate. For our proof,  we show that a syntactically multilinear circuit can in fact be balanced with  the potential function being the number of variables in the sub-circuit rooted at a gate. Our proof of this part involves the machinery of \emph{gate quotients} and \emph{frontier decompositions} developed by Valiant et al. in their original proof, although there are some crucial differences which require some non-trivial (albeit simple) insights. 

One such challenge stems from the fact that while in a homogeneous circuit, the formal degree of any two children of a product gate is the same and equal to the formal degree of the parent, where as the children might depend on very different (even completely disjoint) sets of variables. To get around this, our notion of \emph{frontier} is different from that of Valiant et al~\cite{vsbr83}. In~\cite{vsbr83}, frontier is defined with respect to vertices, whereas we define frontier with respect to edges. As a consequence, our frontier decomposition statements are slightly different from those in~\cite{vsbr83}, although they continue to have a natural semantic meaning. This is detailed in~\autoref{sec:balancing}. 
\paragraph*{ Reduction to depth-$4$ from a balanced circuit. }
In the second part of our proof, we show that any balanced syntactically multilinear circuit of size $s$ computing a polynomial in $n$ variables can be depth reduced to a syntactically multilinear depth-$4$ circuit of size $s^{O(\sqrt{n/\log n})}$. The proof is along the lines of the proof of the analogous statement in the homogeneous (non-multilinear) setting by Chillara et. al.~\cite{CKSV16}. The high level idea of the proof is the following : in a balanced circuit $C$, the polynomial computed at any gate $g$ can be written as a sum of product of \emph{terms}, where the product fan-in is a constant, the sum fan-in is upper bounded by the size of the circuit, and the number of variables in any of the terms is at most half of the number of variables in $g$. Moreover, each of the terms is a polynomial computed by a gate in $C$, so this decomposition can be recursively applied. We apply this decomposition  repeatedly till every term in the sum of products expression of the output depends on at most $t$ variables. We argue that the sum fan-in of this sum of products expression is at most $s^{O(n/t)}$. Now, we expand each of the {terms} (which is a multilinear polynomial) as a sum of multilinear monomials in $t$ variables. Thus, the total size of the $\Sigma\Pi\Sigma\Pi$ circuit obtained is $2^t\cdot s^{O(n/t)}$ which is $s^{O(\sqrt{n/\log s})}$ for $t = \sqrt{n\log s}$. 

In the proof of the analogous statement for homogeneous non-multilinear circuits, at the end of the repeated applications of the decomposition, each of the terms is of degree at most $t$. Thus, a sum of product expansion of each such term has size $\binom{n}{t}$, and so the total size of the $\Sigma\Pi\Sigma\Pi$ circuit obtained is $n^t\cdot s^{O(n/t)}$, which for $s = \poly(n)$ is minimized for $t = \sqrt{n}$ and equals $s^{O(\sqrt{n})}$. This explains the gain in the size obtained by~\autoref{thm:multi-k-ic dr}.
\section{Preliminaries}
In this section, we describe the notion of proof-trees and gate quotients which are crucial to our proof and set up some of the machinery we need for the proof. 
\subsection{Proof-trees and quotients}

\begin{definition}[Proof-trees]\label{defn:proof-tree}
  Let $C$ be an algebraic circuit. For any $u_0 \in C$,  a \emph{proof-tree $T$ rooted at $u_0$} is a subcircuit of $C$ that satifies the following properties:
  \begin{itemize}
  \item the node $u_0\in T$,
  \item if $u \in T$  is a multiplication gate of $C$ with $u = v_1 \times v_2$, then $v_1,v_2$ are also in $T$,
  \item if $u \in T$  is an addition gate of $C$ with $u = v_1 + v_2$, then \emph{exactly one} of $v_1$ or $v_2$ is in $T$.
  \end{itemize}
  Any such sub-circuit computes just a monomial, and this shall be called the  value the proof-tree. Although the proof-tree defined above need not be a tree, it shall unfolded to a tree.

  \medskip

  If $T$ is a proof-tree rooted at $u$, and $v$ is a node that appears on its right-most path, then the tree $T'$ obtained by replacing $v$ only on the right-most path by a leaf labelled $1$ is said to be a \emph{$v$-snipped proof-tree rooted at $u$}. 
\end{definition}

\begin{definition}[$\Var$ operator] For any nodes $u \in C$, we denote by $\Var(u)$ the vector $(d_1,\ldots, d_n) \in \N_{\geq 0}^n$ where $d_i$ is the maximum $x_i$-degree over all proof-trees rooted at $u$.

  Similarly, for any pair of nodes $u,v \in C$, we denote by $\Var(u:v)$ the vector $(d_1,\ldots, d_n)$ where $d_i$ is the maximum $x_i$-degree over all $v$-snipped proof-tree rooted at $u$. 

  \medskip
  
  We shall also define $\abs{(d_1,\ldots, d_n)} = \sum d_i$.   
\end{definition}

\noindent
For a multilinear circuit $C$, note that $\abs{\Var(g)}$ for any gate $g \in C$ is precisely the number of distinct variables in the sub-circuit rooted at $g$. 

\begin{remark}\label{rmk:right heavy}
Throughout this discussion, we will assume that the circuit is \emph{right heavy}. This means that for every multiplication gate, $w = w_{L} \times w_{R}$, $\Var(w_R) \geq \Var(w_L)$. Note that this is without loss of generality, since \emph{left} and \emph{right} are merely labels that we can assign arbitrarily to the children of every gate in the circuit. 
\end{remark}

\begin{definition}[Gate Quotient]\label{def : gate quotient} 
  For every two gates $u, v$ in $C$, the \emph{gate quotient} of $u$ with respect to $v$, denoted by $[u:v]$ is defined inductively as follows. 
\begin{itemize}
\item If $u = v$, then $[u:v] = 1$.
\item If $u = u_1 + u_2$, then $[u:v] = [u_1:v] + [u_2:v]$. 
\item If $u = u_L \times u_R$,  then $[u:v] =  [u_L][u_R:v]$. 
\item If $v$ does not appear in the subcircuit rooted at $u$, then $[u:v] = 0$.\qedhere
\end{itemize}  
\end{definition}

\begin{lemma}\label{lem:gate-quotient-via-proof-trees}
  Let $u,v \in C$. Then, the polynomial $[u]$ is the sum of values of all proof-trees rooted at $u$.
Furthermore, the polynomial $[u:v]$ is the sum of the value of all $v$-snipped proof-trees $T$ rooted at $u$. 
\end{lemma}
The above lemma is almost folklore and a proof of it can be seen in the work of Allender et. al.~\cite{ajmv98}.

\subsection{Syntactic restrictions on proof-trees}

We remark that throughout this paper, by degree, we mean the syntactic or formal degree, which could be much larger than the actual or semantic degree. 
 The following observation records some basic properties of the $\Var$ operator. 

\begin{observation}
  Let $C$ be any algebraic circuit. Then,
  \begin{itemize}
  \item $\Var(u)$ is monotonically non-increasing as $u$ moves towards the leaves.  That is, if $u$ is an ancestor of $v$, then ever coordinate of $\Var(u)$ is at least as large as the corresponding coordinate in $\Var(v)$.

    Similarly, for any fixed $v$, the vector $\Var(u:v)$ is monotonically non-increasing as $u$ moves towards the leaves. 
  \item For any multiplication gate $u = u_1 \times u_2$, we have $\Var(u) =  \Var(u_1) + \Var(u_2)$. Similarly for any $v$, we have $\Var(u:v) = \Var(u_1) + \Var(u_2:v)$. 
  \item For any addition gate $u = u_1 + u_2$, we have $\Var(u) = \max(\Var(u_1), \Var(u_2))$, the coordinate-wise max of the two vectors. Similarly for any $v$, $\Var(u:v) = \max(\Var(u_1:v), \Var(u_2:v))$. 
  \end{itemize}
\end{observation}
\begin{proof}
The proofs immediately follow from the definitions.
\end{proof}

For two vectors $\vecv_1,\vecv_2 \in \N_{\geq 0}^n$, we shall say $\vecv_1 \preceq \vecv_2$ if each coordinate of $\vecv_1$ is at most the corresponding coordinate in $\vecv_2$. 

\begin{observation}\label{obs:rel-between-Var-in-decompositions}
  Suppose $u \in C$ and $w$ is a node in $C$ such that there is some proof-tree rooted at $u$ with $w$ appearing on its rightmost path. Then,
  \[
    \Var(u:w) + \Var(w) \preceq \Var(u).
  \]
  Similarly, suppose $w$ is a node in $C$ such that there is some $v$-proof-tree rooted at $u$ with $w$ appearing on its rightmost path. Then,
  \[
    \Var(u:w) + \Var(w:v) \preceq \Var(u:v).
  \]
\end{observation}
\begin{proof}
The proof is straightforward; we just give the proof of the second equation.
Fix a coordinate $i$.
If $d_i = (\Var(u:w))_i$ then there is some $w$-snipped proof-tree $T_i$ rooted at $u$ whose $x_i$-degree equals $d_i$.
Similarly if $e_i = (\Var(w:v))_i$, then there is some $v$-snipped proof-tree rooted $T_i'$ rooted at $w$ whose $x_i$-degree is $e_i$. Clearly the \emph{gluing} of $T_i$ and $T_i'$ obtained by replacing the snipped vertex $w$ in $T_i$ with the tree $T_i'$ is a $v$-snipped proof-tree rooted at $u$ with $x_i$-degree $d_i + e_i$.
Therefore $d_i + e_i \leq (\Var(u:v))_i$ and the claim follows.
\end{proof}

\begin{definition}[Syntactically multilinear and multi-$k$-ic circuits]\label{defn:multi-k-ic}
  A circuit $C$ is said to be \emph{syntactically multilinear} if $\Var(u) \in \set{0,1}^n$ for all $u\in C$.

  A circuit $C$ is said to be \emph{syntactically multi-$k$-ic} if $\Var(u) \in \set{0,1,\ldots, k}^n$ for all $u\in C$. 
\end{definition}

\section{Frontier edges and quotient}

\begin{definition}[Frontier edges]\label{def:frontier-edges}
  For a circuit $C$, an edge between two gates $g_1,g_2$ (where $g_1$ is the parent) is said to be an \emph{$m$-frontier edge} (for a parameter $m$) if
  \[
    \abs{\Var(g_1)} \geq m\;\text{and}\;\abs{\Var(g_2)} < m.
  \]
  We will use $\mathcal{F}_m^{\times}$ to denote the set of all $m$-frontier edges $(g_1,g_2)$ where $g_1$ is a multiplication gate, and $\mathcal{F}_m^{+}$ to denote those where $g_1$ is an addition gate.

  \medskip

  \noindent
  Furthemore, if $v \in C$ is a fixed gate, we shall say that $(g_1,g_2)$ is an \emph{$m$-frontier edge with respect $v$} if
  \[
    \abs{\Var(g_1:v)} \geq m\;\text{and}\;\abs{\Var(g_2:v)} < m.
  \]
  We will use $\mathcal{F}_{m,v}^{\times}$ to denote the set of all edges $(g_1,g_2)$ that are $m$-frontier edges with respect to $v$ where $g_1$ is a multiplication gate, and $\mathcal{F}_{m,v}^{+}$ to denote those where $g_1$ is an addition gate.
\end{definition}

\section{Decomposition via gate quotients}
In this section, we prove the following lemma, which is the key technical observation needed for our proofs. 
\begin{lemma}\label{lem:vsbr-eqns}
  Let $u,v$ be gates in an algebraic circuit $C$ with $\abs{\Var(u)} \geq m$ and $\abs{\Var(v)} < m$. Then,
  \begin{align}
    [u] & = \sum_{(w,z)\in \mathcal{F}_m^\times} [u:w] \cdot [w_L] \cdot [z] \;+\;  \sum_{(w,z)\in \mathcal{F}_m^+} [u:w] \cdot [z] \label{eqn:vsbr-u}\\
    [u:v] & = \sum_{(w,z)\in \mathcal{F}_{m,v}^\times} [u:w] \cdot [w_L] \cdot [z:v] \;+\;  \sum_{(w,z)\in \mathcal{F}_{m,v}^+} [u:w] \cdot [z:v] \label{eqn:vsbr-uv}
  \end{align}
\end{lemma}

Before giving the formal proof, we shall give an informal sketch  using the concept of proof-trees.  For any $u,v$, we have that $[u:v]$ is the sum of all $v$-snipped proof-trees rooted at $u$. For any proof-tree, since $\abs{\Var(u)} \geq m$ and $\abs{\Var(v)} < m$ and $\Var(\cdot)$ is a monotonically non-increasing function as we move towards the leaves, there must be a unique edge $(w,z) \in \mathcal{F}_{m,v}^\times \union \mathcal{F}_{m,v}^+$ on its right-most path such that $\abs{\Var(w)} \geq m$ and $\abs{\Var(z)} < m$.

  If $(w,z) \in \mathcal{F}_{m,v}^\times$, then  $w = w_L \times z$ is a multiplication gate. Therefore, the sum of the values of all $v$-snipped proof-trees with $w$ (and hence the edge $(w,z)$) on its rightmost path is exactly $[u:w][w:v] = [u:w][w_L][z:v]$.
  
  If $(w,z) \in \mathcal{F}_{m,v}^+$, then  $w = w_1 + z$ is an addition gate. Then, $[u:w]\cdot [w:v]$ is the sum of all $v$-snipped proof-trees with $w$ on its rightmost path and $[u:w][w:v] = [u:w][w_1:v] + [u:w][z:v]$. Each $v$-snipped proof-tree with $w$ on its rightmost path either has $(w,w_1)$  on the rightmost path or $(w,z)$. The term $[u:w][w_1:v]$ is precisely the sum of the values of such\footnote{$v$-snipped proof-trees rooted at $u$ that have $w$ on its rightmost path} proof-trees with $(w,w_1)$ on its rightmost path, and $[u:w][z:v]$ is precisely the sum of the values of those proof-trees with $(w,z)$ on its rightmost path. 

  Since the rightmost path of any $v$-snipped proof-tree rooted at $u$ has a \emph{unique} edge $(w,z)\in \mathcal{F}_{m,v}^\times \union \mathcal{F}_{m,v}^+$, summing over all such potential edges gives
  \[
    [u:v] = \sum_{(w,z)\in \mathcal{F}_{m,v}^\times} [u:w] \cdot [w_L] \cdot [z:v] \;+\;  \sum_{(w,z)\in \mathcal{F}_{m,v}^+} [u:w] \cdot [z:v].\qedhere
  \]  
The proof below is just a formalisation of the above sketch. 

\begin{proof}[Proof of \autoref{lem:vsbr-eqns}]
  The proof shall proceed by induction on the  height of $u$ (leaves are at height $0$). We shall present the proof of \eqref{eqn:vsbr-uv}; the proof of \eqref{eqn:vsbr-u} is analogous.

  \medskip

  \noindent
  {\bf Case $1$:} $u = u_L \times u_R$

  For any $w$, we have that $[u:w] = 1$ if $u = w$, and $[u:w] = [u_1]\cdot [u_2:w]$ whenever $u\neq w$. In particular, since $\abs{\Var(v)} < m \leq \abs{\Var(u)}$ the LHS is $[u:v] = [u_L] \cdot [u_R:v]$.

  If $\abs{\Var(u_R)} \geq m$, then for any $(w,z) \in \F_{m,v}^+$ or $\mathcal{F}_{m,v}^\times$ we have $w \neq u$. Inducting on $u_R$,
  \begin{align*}
    \text{LHS} & = [u_L] \cdot [u_R:v]\\
               & = [u_L] \cdot \inparen{\sum_{(w,z)\in \mathcal{F}_{m,v}^\times} [u_R:w] \cdot [w_L] \cdot [z:v] \;+\;  \sum_{(w,z)\in \mathcal{F}_{m,v}^+} [u_R:w] \cdot [z:v]}\\
               & = \sum_{(w,z)\in \mathcal{F}_{m,v}^\times} [u_L] \cdot [u_R:w] \cdot [w_L] \cdot [z:v] \;+\;  \sum_{(w,z)\in \mathcal{F}_{m,v}^+} [u_L] \cdot [u_R:w] \cdot [z:v]\\
               & =\sum_{(w,z)\in \mathcal{F}_{m,v}^\times} [u:w] \cdot [w_L] \cdot [z:v] \;+\;  \sum_{(w,z)\in \mathcal{F}_{m,v}^+} [u:w] \cdot [z:v] \;=\;\text{RHS}.
  \end{align*}
  On the other hand, if $\abs{\Var(u_R)}<m$ then $[u:w] = 0$ for any $w \neq u$ with $\abs{\Var(w)} \geq m$. Hence,
  \begin{align*}
    \text{RHS} & = \sum_{(w,z)\in \mathcal{F}_{m,v}^\times} [u:w] \cdot [w_L] \cdot [z:v] \;+\;  \sum_{(w,z)\in \mathcal{F}_{m,v}^+} [u:w] \cdot [z:v]\\
               & = [u:u] \cdot [u_L] [u_R:v] = [u:v] = \text{LHS}.
  \end{align*}

  \noindent
  {\bf Case $2$:} $u = u_1 + u_2$

  For any $w$, we have that $[u:w] = 1$ if $u = w$, and $[u:w] = [u_1:w] + [u_2:w]$ whenever $u\neq w$. In particular, since $\abs{\Var(v)} < m \leq \abs{\Var(u)}$ the LHS is $[u:v] = [u_1:v] + [u_2:v]$.

  Since $u$ is a $+$ gate, $(u,u_j) \notin \mathcal{F}_{m,v}^\times$ for any $j$. If  $\abs{\Var(u_j)} < m$ for some $j$, then the edge $(u,u_j) \in \mathcal{F}_{m,v}^{+}$. Hence,
  \begin{align*}
    \text{RHS} & = \sum_{(w,z)\in \mathcal{F}_{m,v}^\times} [u:w] \cdot [w_L] \cdot [z:v] \;+\;  \sum_{(w,z)\in \mathcal{F}_{m,v}^+} [u:w] \cdot [z:v]\\
               & =: \hspace{3cm}T_1\hspace{1.28cm}+\hspace{1.28cm}\quad\quad T_2
  \end{align*}
  In $T_1$, since every $(w,z) \in \mathcal{F}_{m,v}^{\times}$ has $w \neq u$ we have
  \begin{align*}
    T_1 &:= \sum_{(w,z)\in \mathcal{F}_{m,v}^\times} \inparen{\sum_i [u_i:w]} \cdot [w_L] \cdot [z:v]\\
        & = \sum_{(w,z)\in \mathcal{F}_{m,v}^\times} \inparen{\sum_{i:\abs{\Var(u_i)} \geq m} [u_i:w]} \cdot [w_L] \cdot [z:v]\quad\quad\text{(since $[u_j:w] = 0$ if $\abs{\Var(u_j)} < m$)}\\
        & = \sum_{i: \abs{\Var(u_i)} \geq m}\; \sum_{(w,z) \in \mathcal{F}_{m,v}^\times} [u_i:w] \cdot [w_L] \cdot [z:v].
  \end{align*}
  As for the other term, it can be written as
  \begin{align*}
    T_2 & := \sum_{(w,z)\in \mathcal{F}_{m,v}^+} [u:w] \cdot [z:v]\\
        & = \sum_{\substack{(w,z)\in \mathcal{F}_{m,v}^+\\w\neq u}} [u:w] \cdot [z:v] \;+\;\sum_{j:\abs{\Var(u_j)} < m} [u:u] \cdot [u_j:v]\\
        & = \sum_{\substack{(w,z)\in \mathcal{F}_{m,v}^+\\w\neq u}} \inparen{\sum_i [u_i:w]} \cdot [z:v] \;+\; \sum_{j:\abs{\Var(u_j)} < m} [u_j:v]\\
        & = \sum_{\substack{(w,z)\in \mathcal{F}_{m,v}^+\\w\neq u}} \inparen{\sum_{i:\abs{\Var(u_i)}\geq m} [u_i:w]} \cdot [z:v] \;+\; \sum_{j:\abs{\Var(u_j)} < m} [u_j:v] \\
        & = \sum_{i:\abs{\Var(u_i)}\geq m}\;\sum_{(w,z)\in \mathcal{F}_{m,v}^+} [u_i:w] \cdot [z:v] \;\;+\;\; \sum_{j:\abs{\Var(u_j)}<m} [u_j:v].
  \end{align*}
  The last equality holds because $[u_i:u]= 0$. Putting it together,
  \begin{align*}
    \text{RHS} & = T_1 + T_2 \\
               & = \sum_{i:\abs{\Var(u_i)}\geq m} \inparen{\sum_{(w,z) \in \mathcal{F}_{m,v}^\times} [u_i:w] \cdot [w_L] \cdot [z:v] + \sum_{(w,z)\in \mathcal{F}_{m,v}^+} [u_i:w] \cdot [z:v]}\\
               & \quad\quad + \sum_{j:\abs{\Var(u_j)}<m} [u_j:v]\\
               & = \sum_{i:\abs{\Var(u_i)} \geq m} [u_i:v]  + \sum_{j:\abs{\Var(u_j)}<m} [u_j:v] \quad\quad\text{(induction)}\\
               & = [u:v] = \text{LHS}.
  \end{align*}
  \vskip -3em
\end{proof}
\section{Balancing syntactically multilinear circuits}~\label{sec:balancing}
In this section, we prove the following theorem.
\begin{theorem}\label{thm:balancing circuits}
  Suppose $C$ is an algebraic circuit of size $s$. Then, there is a circuit $C'$ of size $\poly(s)$ computing the same polynomial with the following structural properties.
  \begin{itemize}
    \item all addition gates in $C'$ have fan-in $O(s^4)$,
  \item all multiplication gates in $C'$ have fan-in at most $5$,
  \item for any multiplication gate $g\in C'$, any child $h$ of $g$ satisfies $\abs{\Var(h)} \leq \abs{\Var(g)}/2$. 
  \end{itemize}
  Furthermore, if $C$ is syntactically multi-$k$-ic, then so is $C'$. 
\end{theorem}
\begin{proof}
Without loss of generality, we may assume that the circuit is \emph{right-heavy} in the sense that for every multiplication gate $u = u_1 \times u_2$ we have $\abs{\Var(u_2)} \geq \abs{\Var(u_1)}$.  We shall build a new circuit $C'$ that computes all $[u:v]$'s and $[u]$'s for gates $u,v\in C$ using the equations in \autoref{lem:vsbr-eqns}.

  We shall assume inductively that we have already computed all $[w]$'s with $\abs{\Var(w)} < t$ and also all $[w,v]$ with $\abs{\Var(w,v)} < t$. Suppose $u\in C$ such that $\abs{\Var(u)} = t$. Using \eqref{eqn:vsbr-u} from \autoref{lem:vsbr-eqns} with $m = t/2$ we have
  \[
    [u] = \sum_{(w,z)\in \mathcal{F}_m^\times} [u:w] \cdot [w_L] \cdot [z] \;+\;  \sum_{(w,z)\in \mathcal{F}_m^+} [u:w] \cdot [z].    
  \]
  By \autoref{obs:rel-between-Var-in-decompositions}, $\abs{\Var(w)} \geq t/2$ implies that $\abs{\Var(u:w)} \leq t/2$. Furthermore, $\abs{\Var(z)} \leq t/2$ by the choice of the frontier edge and $\abs{\Var(w_L)} \leq t/2$ since $C$ is right-heavy.
This allows us to compute all nodes of the form $[u]$ with $\abs{\Var(u)} \leq t$. 

  \medskip

  \noindent 
  If $u,v \in C$ such that $\abs{\Var(u:v)} = t$. Using \eqref{eqn:vsbr-uv} from \autoref{lem:vsbr-eqns} with $m = t/2$, we have
  \[
    [u:v]  = \sum_{(w,z)\in \mathcal{F}_{m,v}^\times} [u:w] \cdot [w_L] \cdot [z:v] \;+\;  \sum_{(w,z)\in \mathcal{F}_{m,v}^+} [u:w] \cdot [z:v].
  \]
  We can restrict the edges in the RHS to only those edges $(w,z)$ that is present in at least one $v$-snipped proof-tree rooted at $u$ (if not, this edge's contribution to the RHS is zero). Therefore by \autoref{obs:rel-between-Var-in-decompositions}, $\Var(w:v) + \Var(u:w) \preceq \Var(u:v)$ and therefore we have $\abs{\Var(u:w)} \leq t/2$. Furthermore, by the choice of the frontier, we also have $\abs{\Var(z:v)} \leq t/2$. The non-trivial case is $\Var(w_L)$ which could in principle be  large but again $\Var(w_L) \preceq \Var(w:v) \preceq \Var(u:v)$ as any proof-tree rooted $w_L$ is a sub-tree of a $v$-snipped tree rooted at $u$. Since we have already computed all gates $[w]$  with $\Var(w) \leq t$, we can write
  \begin{align*}
    [u:v] & = \sum_{(w,z)\in \mathcal{F}_{m,v}^\times} [u:w] \cdot [w_L] \cdot [z:v] \;+\;  \sum_{(w,z)\in \mathcal{F}_{m,v}^+} [u:w] \cdot [z:v]\\
          & = \sum_{(w,z)\in \mathcal{F}_{m,v}^\times} [u:w] \cdot \inparen{\sum_{(p,q)\in \mathcal{F}_{m_{w}}^\times} [w_L:p] \cdot [p_L] \cdot [q] \;+\;  \sum_{(p,q)\in \mathcal{F}_{m_{w}}^+} [w_L:p] \cdot [q]}  \cdot [z:v]\\
          & \qquad + \sum_{(w,z)\in \mathcal{F}_{m,v}^+} [u:w] \cdot [z:v],
  \end{align*}
  where $m_{w} = \Var(w_L)/2$. \\

  The required structural properties of $C'$ are readily seen from the above construction. 
\end{proof}

\section{Reduction to depth four from balanced circuits}
We now show how to reduce a balanced circuit to a depth-$4$ circuit. This would complete the proof of our main theorem. We shall use the notation $\vSPSP{t}$ to refer to $\Sigma\Pi\Sigma\Pi$ circuits computing polynmomials of the form
\[
  F = \sum_{i} \prod_j Q_{ij},
\]
with $\abs{\Var(Q_{ij})} \leq t$.

The proof of this part follows the outline of a similar  argument in Chillara et. al.~\cite{CKSV16} of reducing to depth-$4$ from a balanced circuit. However, there are some differences: our potential is $\abs{\Var(\;\cdot\;)}$ and not the degree (as is usually the case).  Since this potential function also falls as we go from a sum $(+)$ gate to its children, we need one more simple observation in our argument to bound the number of steps in the recursion in the proof. We now provide the details.

\begin{lemma}\label{lem : from balanced circuit to depth 4}
  Let $C$ be a multi-$k$-ic circuit of size $s$ such that every multiplication gate  $g$ in $C$ fan-in at most $5$ and for every child $h$ of $g$ in $C$, $\Var(h) \leq \Var(g)/2$.

  Then, for any positive integer $0\leq t \leq kn$, there is an equivalent multi-$k$-ic $\vSPSP{t}$ circuit $C'$ that computes the same polynomial, with the following properties:
  \begin{itemize}
  \item the top fan-in of $C'$ is at most $s^{O(kn/t)}$,
  \item the size of $C'$ is at most $2^{kt} \cdot s^{O(kn/t)}$,
  \item each of the $(+)$-gates closer to the leaves compute polynomials that computed by gates in $C$. 
  \end{itemize}
\end{lemma}
\begin{proof}
Since $C$ is balanced, with product fan-in at most $5$, every gate $g$ in $C$ can be written as
\begin{equation}\label{eqn : single iteration}
g = \sum_{i = 1}^s \prod_{j = 1}^5 g_{i,j} \, ,
\end{equation}
where each $g_{i,j}$ is also computed by a gate in the circuit $C$, $|\Var(g_{i,j})| \leq |\Var(g)|/2$. With this notation, \eqref{eqn : single iteration} applied on the root of $C$ says that  $C$, which is a syntactically multi-$k$-ic circuit, can be trivially written as a $\vSPSP{kn/2}$. A natural idea would be to apply \eqref{eqn : single iteration} on the $g_{i,j}$'s until we get a $\vSPSP{t}$ circuit. All that is needed is to bound the number of summands (or the top fan-in of the resulting $\vSPSP{t}$ circuit) at the end of this process. Observe that for every $i \in \{1, 2, \ldots, s\}$,  we could have that $\abs{\Var(\prod_{j = 1}^5 g_{i,j})}$ is \emph{much} smaller than $\abs{\Var(g)}$ itself. 

We will view the process as a tree in the natural way. The root of the tree corresponds to the root of the circuit, and all other nodes in the tree correspond to products of addition gates in $C$. The children of a node in the tree correspond to the summands in the sum of product representation of that node obtained by expanding one of its factors according to \eqref{eqn : single iteration}. The leaves of this tree are products of addition gates $\prod g_{i}'$ such that $\abs{\Var(g_i')} \leq t$ for each factor $g_i'$.
The tree has a branching factor of at most $s$, hence it suffices to get a bound on the depth of the tree to get a bound on the number of leaves which would be the top fan-in of the $\vSPSP{t}$ representation. 

Let $g\prod_{\ell} w_{\ell}$ be an internal node in the tree with $\abs{\Var(g)} > t$. After applying \eqref{eqn : single iteration} on $g$, we get   
\[
g\left(\prod_{\ell} w_{\ell}\right) = \sum_{i = 1}^s \left(\prod_{j = 1}^5 g_{i,j} \cdot \prod_{\ell}w_{\ell}\right).
\]
We now consider two cases. 
\begin{itemize}
\item  {\bf $\abs{\Var\left(\prod_{j = 1}^5 g_{i,j}\right)} < 3t/4$ :} In this case, $\abs{\Var\left(\prod_{j = 1}^5 g_{i,j} \cdot \prod_{\ell}w_{\ell}\right)} \leq \abs{\Var\left( g\cdot \prod_{\ell}w_{\ell} \right)} - t/4$.

\item {\bf $\abs{\Var\left(\prod_{j = 1}^5 g_{i,j}\right)} \geq 3t/4$ :}
  Since $\Var(g) \succeq \Var(g_{i,1} \cdots g_{i,5}) = \Var(g_{i,1}) + \cdots + \Var(g_{i,5})$ and $\abs{\Var(g_{i,j})} \leq t/2$, it follows that the number of factors $h$ in $\prod_{j = 1}^5 g_{i,j} \cdot \prod_{\ell}w_{\ell}$ with $\abs{\Var(h)} \geq t/16$ is at least one more than the number of such factors in $g\cdot \prod_{\ell}w_{\ell}$. This is because besides the factor $g_{i,j}$ with largest $\abs{\Var(g_{i,j})}$, the other four factors together must contribute at least $(3t/4) - (t/2) = (t/4)$ to $\abs{\Var(g_{i,1}\cdots g_{i,5})}$ and hence at least one of them must have $\abs{\Var(g_{i,k})} \geq t/16$. 
\end{itemize}
Thus, in any edge of the tree, either $\abs{\Var(\;\cdot\;)}$ decreases by $t/4$ or the number of factors with $\abs{\Var(\; \cdot \;)} \geq t/16$ increases by one. The root node $g_0$ has $\abs{\Var(g_0)} \leq kn$.  Hence, the depth of the tree is bounded by $(16 + 4)(kn/t) = O(nk/t)$. Therefore, $C$ can be computed by a syntactically multi-$k$-ic $\vSPSP{t}$ circuit of top fan-in at most $s^{O(nk/t)}$.  

To get the bound on the overall size of the $\vSPSP{t}$ circuit, we need to bound the sparsity of the polynomials computed by  bottom two layers. Note that if $\Var(f) = (d_1,\ldots, d_n)$, then $f$ can have at most $\prod (1 + d_i)$ monomials. Since $2^x \geq 1 + x$ for all positive integers $x$, it follows that $\abs{\Var(f)} \leq t$ implies that $f$ has at most $2^t$ monomials.  Therefore, the total size of the $\vSPSP{t}$ circuit is $2^t \cdot s^{O(kn/t)} = 2^{O\left(t + \frac{kn\log s}{t}\right)}$. 
\end{proof}

\noindent
From~\autoref{thm:balancing circuits} and setting $t = \sqrt{kn \log s}$ in \autoref{lem : from balanced circuit to depth 4}, we get \autoref{thm:multi-k-ic dr} restated below. 

\multiKicDR*

\subsection{Reduction to higher depths}

We now prove~\autoref{thm:multi-k-ic dr Delta} which shows that similar savings can be obtained in depth reductions to larger depth. 

\multiKicDRDelta*

\begin{proof}[Proof of~\autoref{thm:multi-k-ic dr Delta}]
  We shall assume, without loss of generality, that the circuit $C$ is balanced (by applying \autoref{thm:balancing circuits} if necessary). 
  The proof follows via repeated applications of \autoref{lem : from balanced circuit to depth 4}.

  Applying \autoref{lem : from balanced circuit to depth 4} with $t = nk /(nk/\log s)^{1/\Delta}$, we obtain a $\vSPSP{t}$ circuit $C'$ of the form
  \[
    C' = \sum_{i=1}^{s'} \prod_j g_{ij},
  \]
  with $s' = s^{O\inparen{(kn/\log s)^{1/\Delta}}}$ and $\abs{\Var(g_{ij})}\leq t$ for all $i,j$. Furthermore, since each $g_{ij}$ being a polynomial computed by a gate in $C$, they are computable by multi-$k$-ic circuits of size at most $s$. By induction, each $g_{ij}$ has a multi-$k$-ic $(\Sigma\Pi)^{\Delta - 1}$ circuit of size at most
  \[
    s^{O\inparen{(\Delta - 1) \cdot (t/\log s)^{1/(\Delta - 1)}}} = s^{O\inparen{(\Delta - 1) \cdot (nk/\log s)^{1/\Delta}}}. 
  \]
  Replacing each $g_{ij}$ by this circuit, we obtain a $(\Sigma\Pi)^{\Delta}$ circuit of size at most
  \[
    s' \cdot s^{O\inparen{(\Delta - 1) \cdot (nk/\log s)^{1/\Delta}}} = s^{O\inparen{\Delta \cdot (kn/\log s)^{1/\Delta}}}. 
    \]
    \vskip -2em    
\end{proof}
\section{Open problems}
The most interesting question that comes out of this work is to prove a lower bound of $n^{\omega(\sqrt{n/\log n})}$ for syntactically multilinear circuits of depth-$4$ for an explicit polynomial. A natural and first approach to this could be to understand if the shifted partials based methods can prove a lower a lower bound of $n^{\Omega(\sqrt{d})}$ for homogeneous depth-$4$ circuits for a polynomial family with degree $d = \omega(n/\log n)$. 

Another question of interest would be to understand the \emph{correct} exponent for the depth reduction results to depth-$4$ (and also to higher depth) for various regimes of the degree $d$. From~\cite{KS14}, we know that for $d = O(n^{\epsilon})$ for a small enough constant $\epsilon$, $\sqrt{d}$ is the correct exponent, whereas for $d$ being nearly $n$, the results in this paper and those of Raz and Yehudayoff~\cite{RY09} show that the correct exponent is $\sqrt{n/\log n}$. But we do not understand this phenomenon for other values of $d$. 
\section*{Acknowledgements}
We are deeply thankful to Ben Rossman, who pointed us towards this question, and for many stimulating discussions at various stages of this work. We also thank Shubhangi Saraf, Amir Shpilka and Ben Lee Volk for many helpful conversations. 

Mrinal is also thankful to Prahladh Harsha for accommodating him in his apartment for a part of the visit to TIFR, where a part of this paper was written. 

\bibliographystyle{customurlbst/alphaurlpp}
\bibliography{bib/references,bib/crossref}

\end{document}